\documentclass[11pt]{article}

\usepackage{sgamevar}
\usepackage{amsmath,amsthm,amsfonts,amssymb}
\usepackage{natbib}
\usepackage{dsfont}

\usepackage[letterpaper,top=1.5in,bottom=1in,left=1.25in,right=1.25in]{geometry}

\usepackage{todonotes}

\newcommand{\citeeg}[1]{\citep[e.g.,][]{#1}}

\usepackage[primary]{xqting}

\newcommand{\abs}[1]{\lvert{}#1\rvert{}}

\DeclareMathOperator*{\argmin}{arg\,min}
\DeclareMathOperator*{\argmax}{arg\,max}

\newcommand{\N}{\mathds{N}}
\providecommand{\R}{\mathds{R}}\renewcommand{\R}{\mathds{R}}

\newcommand{\G}{\mathcal{G}}

\newcommand{\Ind}[1]{\mathds{I}\left[#1\right]}

\newcommand{\fk}[1][k]{f_i^{(#1)}}

\theoremstyle{plain}
\newtheorem{theorem}{Theorem}
\newtheorem{corollary}{Corollary}
\newtheorem{lemma}{Lemma}
\newtheorem{proposition}{Proposition}

\theoremstyle{definition}
\newtheorem{definition}{Definition}
\newtheorem{example}{Example}

\begin{document}

\title{A Formal Separation Between\\Strategic and Nonstrategic Behavior}

\author{
James R.\ Wright\thanks{Department of Computing Science and Alberta Machine Intelligence Institute, University of Alberta. E-mail: james.wright@ualberta.ca}
\and
Kevin Leyton-Brown\thanks{Department of Computer Science, University of British Columbia. E-mail: kevinlb@ubc.ca}
}
\date{}


\maketitle

\begin{abstract}
It is common to make a distinction between ``strategic'' behavior and other forms of intentional but ``nonstrategic'' behavior: typically, that strategic agents respond to beliefs about the behavior of other agents while nonstrategic agents do not. However, a crisp boundary between these concepts has proven elusive. 
This problem is pervasive throughout the game theoretic literature on bounded rationality and particularly critical in parts of the behavioral game theory literature that make an explicit distinction between the behavior of ``nonstrategic'' level-$0$ agents and ``strategic'' higher-level agents (e.g., the level-$k$ and cognitive hierarchy models). Overall, work discussing bounded rationality rarely gives clear guidance on how the rationality of nonstrategic agents must be bounded, instead typically just singling out specific decision rules (e.g., randomizing uniformly, playing toward the best case, optimizing the worst case) and informally asserting that they are nonstrategic.
In this work, we propose a new, formal characterization of nonstrategic behavior. Our main contribution is to show that it satisfies two properties: (1) it is general enough to capture all purportedly ``nonstrategic'' decision rules of which we are aware in the behavioral game theory literature; (2) behavior that obeys our characterization is distinct from strategic behavior in a precise sense. 
\end{abstract}

\textbf{Keywords:} game theory; behavioral game theory; bounded rationality; cognitive models; cognitive hierarchy; level-k

\section{Introduction}

A common assumption in the game theoretic literature is that agents are perfect optimizers who form correct, explicitly probabilistic beliefs and best respond to those beliefs. The behavior of such agents is commonly said to be \emph{strategic}; indeed, in the early days of game theory, the term ``strategic'' was used as a synonym for perfect rationality \citeeg{bernheim84,pearce84}.
%
At the other extreme, we might assume that agents neglect to model other agents at all, following some fixed rule like playing a specific default action or uniformly randomizing across all available actions.
The behavior of such agents is commonly said to be \emph{nonstrategic}.
%
Things get muddier in between these extremes. 
Human players are clearly not perfect optimizers; e.g., nobody knows the Nash equilibrium strategy for chess. 
However, at least some of us surely do reason about the behavior and beliefs of other agents with whom we interact. 
The literature generally also calls such ``boundedly rational'' behavior strategic, even when the behavior is inconsistent with perfect optimization; intuitively, the dividing line is generally taken to be the question of whether agents model other agents and their incentives when deciding how to act.
More formally, the term ``strategic'' is generally used to describe agents who act to maximize their own utilities based on explicit probabilistic beliefs about the actions of other agents \citep{roth02,li03,babaioff04,lee14,gerding11,ghosh12,grabisch17},
and the term ``nonstrategic'' is generally used to describe agents who follow some fixed, known decision rule \citep{sandholm01,airiau03,li03,lee14,gerding11,grabisch17}.

Being able to make a sharp distinction between strategic and nonstrategic behavior matters particularly in the context of  
a prominent family of predictive models from behavioral game theory that describe  \emph{iterative strategic reasoning}, including the level-$k$ \citep{nagel95,costagomes01,crawford10}, cognitive hierarchy \citep{camerer04cognitive}, and quantal cognitive hierarchy models \citep[e.g.,][]{wright17}.
In all of these models, some agents are strategic in a boundedly rational sense, performing a finite number of recursive steps of reasoning about the behavior of other agents, ultimately terminating in reasoning about so-called level-$0$ agents, who are assumed to be nonstrategic.
%
Level-$0$ behavior is frequently defined simply as uniform randomization.
However, the predictive performance of such models can often be substantially improved by allowing for richer level-$0$ specifications \citep{wright-ec,wright-jair}.
For example, one could specify that level-$0$ agents act to maximize the utility of their worst case (maxmin) or of their best case (maxmax).  Because these rules require only the acting agent's utilities as inputs (and not those of any of the other agents), it seems clear that they are nonstrategic.  
%
But just because a proposed level-$0$ behavior can be written without reference to beliefs about other agents' strategies, we cannot conclude that there does not exist another, equivalent way of writing it that does depend on such beliefs.  For example, the maxmax rule just described can also be expressed as a best response to the belief that the opposing agents will play actions that make it possible for the acting agent to achieve their best-case outcome.
%
Things get even worse if one aspires to learn the level-$0$ specification directly from data, effectively optimizing over a space of specifications
\citep{hartford16}: the task now becomes reassuring a skeptic that no point in this space corresponds to behavior that could somehow be rewritten in strategic terms.

This paper defines minimal conditions that we argue must be satisfied by any strategic behavioral model.
We first define several standard solution concepts, and prove that they are strategic according to our definition.
We then leverage this definition to characterize a broad family of ``nonstrategic'' decision rules---called the \emph{elementary} behavioral models---and show that they deserve the name: i.e., that no rule in this class can represent strategic reasoning.
Our proposed characterization is a structural notion: it restricts the information that agents are permitted to use by restricting them to summarize all outcomes into a single number before performing their reasoning.
Finally, we consider the effects of combining elementary models to construct more complicated behavioral models.
Convex combinations of elementary models are also nonstrategic when the elementary models are compatible in a precise sense.
Overall, our results are important because they distinguish strategic from nonstrategic behavioral models via formal mathematical criteria, rather than relying on the intuitive sense that a model ``depends on'' an explicit model of an opponent's behavior. This makes it possible, for example, to introduce a rich, highly parameterized level-$0$ specification into an iterative strategic model while guaranteeing that there is no way of instantiating the level-$0$ specification to produce strategic behavior.

\section{Background}
\label{sec:background}
We begin by briefly defining our formal framework and notation, discussing normal-form games, solution concepts, and behavioral models.

\subsection{Normal-Form Games}

A normal-form game $G$ is defined by a tuple $(N,A,u)$, where
$N=\{1,\ldots,n\}$ is a finite set of \emph{agents}; $A=A_1 \times \ldots \times A_n$ is the set of possible \emph{action profiles}; $A_i$ is the finite set of \emph{actions} available to agent $i$; and $u=\{u_i\}_{i \in N}$ is a set of \emph{utility functions} $u_i:A \to \R$, each of which maps from an action profile to a utility for agent $i$.
Agents may also randomize over their actions. It is standard in the literature to call such randomization a \emph{mixed strategy}; however, for our purposes this terminology will be confusing, since it would lead us to discuss the strategies of nonstrategic agents. We thus instead adopt the somewhat nonstandard terminology \emph{behavior} for this concept. We denote the set of agent~$i$'s possible behaviors by $S_i = \Delta^{|A_i|}$, and the set of possible \emph{behavior profiles} by $S = S_1 \times \ldots \times S_n$, where $\Delta^k$ denotes the standard $k$-simplex (the set $\left\{\theta_0 + \dots +\theta_k ~\bigg|~ \sum_{i=0}^{k} \theta_i=1 \mbox{ and } \theta_i \ge 0 \mbox{ for all } i\right\}$), and hence $\Delta^{|X|}$ is the set of probability distributions over a finite set $X$.
Overloading notation, we represent the expected utility to agent~$i$
of a behavior profile $s \in S$ by $u_i(s)$.
We use the notation $s_{-i}$ to refer to the behavior profile of all agents except $i$, and $(s_i, s_{-i})$ to represent a full behavior profile.

\subsection{Solution Concepts}
\label{sec:solution-concept}

A \emph{solution concept} is a mapping from a game $G$ to a behavior profile (or set of behavior profiles) that satisfies some criteria. We will primarily be concerned with these solution concepts as formalizations of \emph{strategic behavior} in games.

The foundational solution concept in game theory, and the most commonly used, is the Nash equilibrium.
\begin{definition}[Nash equilibrium]
    Let $BR_i(s_{-i}) = \argmax_{a_i \in A_i} u_i(a_i, s_{-i})$
    denote the set of agent~$i$'s \emph{best responses} to a behavior profile $s_{-i} \in S_{-i}$.
    A Nash equilibrium is a behavior profile in which every agent simultaneously best responds to all the other agents.  Formally, $s^* \in S$ is a \emph{Nash equilibrium} if
    $\forall i \in N, a_i \in A_i:\: s_i^*(a_i) > 0 \implies a_i \in BR_i(s_{-i}^*).$
\end{definition}

One important idea is that people become more likely to make errors as the cost of making those errors decreases.
This can be modeled by assuming that agents best respond \emph{quantally}, rather than via strict maximization.
A quantal response plays actions with high expected utility with high probability, and actions with low expected utility with lower probability.
An equilibrium in which agents quantally respond to each other, rather than best responding to each other, is called a quantal response equilibrium \citep{mckelvey95}.
\begin{definition}[Quantal best response]
  A \emph{(logit) quantal best response} $QBR_i(s_{-i}; \lambda,G)$ by agent $i$ to $s_{-i}$ in game $G$ is a behavior $s_i$ such that
  \begin{equation}
    s_i(a_i) =  \frac{\exp[\lambda\cdot u_i(a_i, s_{-i})]}{\sum_{a'_i}\exp[\lambda \cdot u_i(a'_i, s_{-i})]},\label{eq:qbr}
  \end{equation}
  where $\lambda$ (the \emph{precision} parameter) indicates how sensitive agents are to utility differences.
  When $\lambda=0$, quantal best response is equivalent to uniform randomization.
  As $\lambda\rightarrow\infty$,
  quantal best response corresponds to best response in the sense that actions that are not best responses are played with probability that approaches zero; i.e., for all $a_i \notin BR_i(s_{-i})$, it is that case that $\lim_{\lambda\to\infty} QBR_i(s_{-i};\lambda)(a_i) = 0$.
\end{definition}
\begin{definition}[QRE]
  A \emph{quantal response equilibrium} with precision
  $\lambda$ is a behavior profile $s^*$ in which every agent's
  behavior is a quantal best response to the behaviors of the other
  agents; i.e., for all agents $i$,
  $s^*_i = QBR^G_i(s^*_{-i};\lambda).$
\end{definition}
In general, a game can have multiple QREs with a given precision.
However, when using QRE for predictions, it is common to select a particular equilibrium for a given $\lambda$ that lies on a one-dimensional manifold in the joint space of strategies and precisions that starts from the uniform strategy at $\lambda=0$.  Note that, although quantal best response approaches best response as precision goes to infinity, there nevertheless exist Nash equilibria that are not the limits of any sequence of QREs.

\subsection{Models of Agent Behavior}
\label{sec:models}

\newcommand{\AGi}{\ensuremath{A}_{G,i}}

We now turn to what we term \emph{behavioral models}, functions that return a probability distribution over a single agent's action space for every given game.  We will often refer to this distribution informally as the model's ``behavior''.
Unlike a solution concept, which represents a set of joint strategies that are consistent with some criterion, a behavioral model represents a prediction of a single agent's actions.  Profiles of behavioral models can thus be seen as solution concepts that always encode a single product distribution over a given set of individual behaviors.  In a later section we consider what can be said about the profiles of behavioral models induced by existing solution concepts.

Before we can define behavioral models, we must introduce some basic notation. Let:
\begin{itemize}
    \item $\G$ denote the space of all finite normal-form games;
    \item $\R^* = \bigcup_{k=1}^\infty \R^k$ denote the space of all finite vectors; 
    \item $(\R^*)^* = \bigcup_{k=1}^\infty (\R^*)^k$ denote the space of all finite-sized, finite-dimensional tensors; 
    \item $\Delta^* = \bigcup_{k=1}^\infty \Delta^k$ denote the space of all finite standard simplices; and 
    \item $\AGi$ denote player $i$'s action space in game $G$.
\end{itemize}

\begin{definition}[Behavioral models]
    A \emph{behavioral model} is a function $f_i:\G \to \Delta^*$; $f_i(G) \in \Delta^{|\AGi|}$ for all games $G$.
    We use a function name with no agent subscript, such as $f$, to denote a \emph{profile of behavioral models} with one function $f_i$ for each agent.
    We write $f(G)$ to denote the behavior profile that results from applying each $f_i$ to $G$.
\end{definition}

Much work in behavioral game theory proposes behavioral models rather than solution concepts (though the formal definition of a behavioral model is our own).\footnote{\label{fn:no-correlation}
Since a behavioral model returns a distribution over only a single agent's actions, a profile of behavioral models cannot represent correlations between agents behaviors; e.g., no non-Nash correlated equilibrium can be represented by a profile of behavioral models.  In this work we focus on understanding and predicting how individual agents reason about the other agents; we leave the modeling of correlations between randomized actions of agents to future work.}

One key idea from that literature is that humans can only perform a limited number of steps of strategic reasoning, or equivalently that they only reason about higher-order beliefs up to some fixed, maximum order.  

We begin with the so-called \emph{level-$k$} model \citep{nagel95,costagomes01}. Unlike Nash equilibrium and quantal response equilibrium, both of which describe fixed points, the level-$k$ model is computed via a finite number of best response calculations. Each agent $i$ is associated with a level $k_i \in \N$, corresponding to the number of steps of reasoning the agent is able to perform.  A level-$0$ agent plays nonstrategically (i.e., without reasoning about its opponent); a level-$k$ agent (for $k\ge1$) best responds to the belief that all other agents are level-$(k-1)$.
The level-$k$ model implies a distribution over play for all agents when combined with a distribution $D \in \Delta^*$ over levels.
\begin{definition}[Level-$k$ prediction]
    Fix a distribution $D \in \Delta^*$ over levels and a level-$0$ behavior $s^0 \in S$.
    Then the \emph{level-$k$ behavior} for an agent $i$ is defined as
    \[s^k_i(a_i) \propto \Ind{a_i \in BR_i(s^{k-1}_{-i})},
    \]
    where $\Ind{\cdot}$ is the indicator function that returns 1 when its argument is true and 0 otherwise.
    The \emph{level-$k$ prediction} $\pi^{Lk} \in S$ for a game $G$ is the average of the behavior of the level-$k$ strategies weighted by the frequency of the levels,
    $\pi^{Lk}_i(a_i) = \sum_{k=0}^\infty D(k)s^k_i(a_i).$
\end{definition}

Cognitive hierarchy \citep{camerer04cognitive} is a very similar model in which agents respond to the distribution of lower-level agents, rather than believing that every agent performs exactly one step less of reasoning.
\begin{definition}[Cognitive hierarchy prediction]
    Fix a distribution $D \in \Delta^*$ over levels and a level-$0$ behavior $s^0 \in S$.
    Then the \emph{level-$k$ hierarchical behavior} for an agent $i$ is
    \[\pi^k_i(a_i) \propto \Ind{a_i \in BR_i(\pi^{0:k-1}_{-i})},
    \]
    where $\pi^0 = s^0$ and $\pi^{0:k-1}_i(a_i) = \sum_{m=0}^{k-1} D(m)\pi^m_i(a_i)$.
    The \emph{cognitive hierarchy prediction} is again the average of the level-$k$ hierarchical strategies weighted by the frequencies of the levels,
    $\pi^{CH}_i(a_i) = \sum_{k=0}^\infty D(k)\pi^k_i(a_i).$
\end{definition}

As we did with quantal response equilibrium, it is possible to generalize these iterative solution concepts by basing agents' behavior on quantal best responses 
rather than best responses.  The resulting models are called \emph{quantal level-$k$} and \emph{quantal cognitive hierarchy} (e.g., \citealp{stahl94}; \citealp{wright17}).

All of the iterative solution concepts described above rely on the specification of a nonstrategic level-$0$ behavior.
This need is a crucial motivation for the current paper, in which we explore what behaviors can be candidates for this specification.

\section{Strategic Behavioral Models}
\label{sec:strategic-models}

As discussed in the introduction, there is general qualitative agreement in the literature that strategic agents act to maximize their own utilities based on explicit probabilistic beliefs about the actions of other agents. Our ultimate goal is to characterize behavioral models that are unambiguously nonstrategic; thus, to strengthen our results, we adopt a somewhat more expansive notion of strategic behavior. Specifically, we define an agent as \emph{minimally strategic} if they satisfy two conditions, which we call (1)~other responsiveness and (2)~dominance responsiveness. These conditions require that the agent chooses actions both (1)~with at least some dependence on others' payoffs; and (2)~with at least some concern for their own payoffs. 

The key feature of strategic agents is that they take account of the incentives of other agents when choosing their own actions.  To capture this intuition via the weakest possible necessary condition, we say that an agent is other responsive if their behavior is ever influenced by changes (only) to the utilities of other agents.

\begin{definition}[Other responsiveness]
    A behavioral model $f_i$ is \emph{other responsive} if there exists a pair of games $G=(N,A,u)$ and $G'=(N,A,u')$ such that $u_i(a) = u'_i(a)$ for all $a \in A$, but $f_i(G) \ne f_i(G')$.
\end{definition}

It is also traditional to assume that agents always act to maximize their expected utilities. This assumption is too strong for our purposes; for example, we want to allow for deviations from perfect utility maximization such as quantal best response. However, it does not seem reasonable to call an agent strategic if they pay no attention whatsoever to their own payoffs.
%
A natural analogue to other responsiveness might seem to be \emph{payoff responsiveness} \citep[as in][]{goeree2005regular}, where agents are assumed to play higher expected-utility actions with higher probability.  Computing the expected utility of a player's action requires a belief over the actions of the other players, so we cannot use this notion without fixing the specific beliefs of the agent.
We thus introduce a concept that we call \emph{dominance responsiveness},
a sense in which an agent might show concern for their own payoffs that is weaker in two ways.
First, it requires that agents with an action that has higher expected utility \emph{regardless of the actions of the other players} play that action with higher probability than other actions.  Second, it only requires this when the action's utility is \emph{sufficiently} higher; agents are not required to respond to arbitrarily small differences in utility.

\begin{definition}[Dominance responsiveness]
    \label{def:dominance-responsive}
    Fix a game $G=(N,A,u)$ in which action $a_i^+ \in A_i$ is strictly dominant; that is,
    $u_i(a_i^+, a_{-i}) > u_i(a_i, a_{-i})$ for all $a_{-i} \in A_{-i}$ and $a_i \ne a_i^+$.
    A behavioral model $f_i$ respects dominance in $G$ if $f_i(G)(a_i^+) > f_i(G)(a_i)$ for all $a_i \ne a_i^+$.
    We say that $f_i$ is \emph{$\zeta$-dominance responsive} if it respects dominance in all games in which 
    $u_i(a_i^+, a_{-i}) > u_i(a_i, a_{-i}) + \zeta$ for all $a_i \ne a_i^+$ and $a_{-i} \in A_{-i}$.
    A behavioral model is \emph{dominance responsive} if it is $\zeta$-dominance responsive for some $\zeta > 0$.
\end{definition}

Other responsiveness requires only that a behavioral model sometimes change its behavior in response to changes in the other players' utilities.  However, dominance responsiveness requires its behavior to change \emph{whenever} the player's own utilities satisfy the condition that one action is sufficiently dominant.
We impose this stronger requirement as a way to strike a balance between two objectives.
On the one hand, we want to require that a strategic model in some way aims to achieve high utility outcomes.
However, we do not want to assume a specific mechanism for seeking higher utilities;
in particular, we do not want to require optimization with respect to probabilistic beliefs, and we do not want to require that a player ``notice'' arbitrarily small changes in its own utilities.

Behavioral models that are both other responsive and dominance responsive satisfy a very low bar for strategic behavior.  That means that models that fail to satisfy either one or both conditions are unambiguously nonstrategic; we refer to such models as \emph{strongly nonstrategic}.

\begin{definition}[Strongly nonstrategic behavioral model]
    A behavioral model $f_i$ is \emph{minimally strategic} if it is both other responsive and dominance responsive.
    Conversely, a behavioral model that is not minimally strategic is \emph{strongly nonstrategic}.
\end{definition}


\section{Existing Solution Concepts are Minimally Strategic}
\label{sec:analysis-existing}

We now demonstrate that our definition of minimally strategic behavioral models does more than describe qualitative patterns of behavior that have been called ``strategic'' in the past: it also formally captures the predictions of various solution concepts both from classical game theory and from behavioral game theory (Nash equilibrium, quantal response equilibrium, level-$k$, cognitive hierarchy, and quantal cognitive hierarchy).
 
We first show that any model defined as either a best response or a quantal best response to a profile of dominance responsive behavioral models is itself dominance responsive.

\begin{lemma}
    \label{lem:qbr-DR}
    For any $\lambda>0$ and profile of behavioral models $f_{-i}$, the behavioral model
    $q_i(G) = QBR_i(f_{-i}(G); \lambda, G)$   
    is dominance responsive.
\end{lemma}
\begin{proof}
    Fix $G=(N,A,u)$ with some strictly dominant action $a^+_i \in A_i$.
    For any $a_i \in A_i$,
    $q_i(G)(a_i) \propto \exp[u_i(a_i^+, f_{-i}(G))]$.
    Since $u_i(a_i^+, f_{-i}(G)) > u_i(a_i, f_{-i}(G))$ for any $a_i \ne a_i^+$, we have $q_i(a_i^+) > q_i(a_i)$.
    That is, $q_i$ always plays a strictly dominant action with higher probability than any other action, regardless of the difference in utilities.  Thus, $q_i$ is $\zeta$-dominance responsive for any $\zeta>0$.
\end{proof}

\begin{lemma}
    \label{lem:br-DR}
    For any profile of behavioral models $f_{-i}$, any behavioral model $b_i$ that satisfies
    \[ a_i \notin BR_i(f_{-i}(G)) \implies b_i(G)(a_i) = 0 \]
    is dominance responsive.
\end{lemma}
\begin{proof}
    For any action $a^+_i \in A_i$, if $u_i(a^+_i, a_{-i}) > u_i(a_i,a_{-i})$ for all $a_i\ne a^+_i\in A_i$ and $a_{-i} \in A_{-i}$, then by assumption $b_i(G)(a_i)=0$ for all $a_i \ne a^+_i$, and hence $b_i(G)(a^+_i)=1 > 0$.  Thus, $b_i$ is also $\zeta$-dominance responsive for any $\zeta > 0$.
\end{proof}

We now show that quantal best response to any profile of dominance responsive behavioral models is minimally strategic.
From this result, it will immediately follow that Nash equilibrium, quantal response equilibrium, level-$k$, cognitive hierarchy, and quantal cognitive hierarchy are minimally strategic.

\begin{lemma}
    \label{lem:strategic-QBR}
    For any $\lambda>0$ and dominance responsive profile of behavioral models $f_{-i}$, the behavioral models $q_i$ and $b_i$ are both minimally strategic, where
    \[q_i(G) = QBR_i(f_{-i}(G); \lambda, G)\]
    and $b_i$ satisfies
    \[ a_i \notin BR_i(f_{-i}(G)) \implies b_i(G)(a_i) = 0. \]
    \begin{proof}
    Both models are dominance responsive by Lemmas~\ref{lem:qbr-DR} and~\ref{lem:br-DR} respectively.
    It remains only to show that both models are other responsive.

    Consider the following 2-player games, in which the opponent $j$ is $\zeta$-dominance responsive.
    \begin{center}
            \begin{game}{3}{2}[$G_1$]
                 \> $L$ \> $R$ \\
             $U$ \> $1,1+\zeta$  \> $0,0$ \\
             $D$ \> $0,1+\zeta$  \> $1,0$
            \end{game}%
            \hspace{1cm}%
            \begin{game}{3}{2}[$G_2$]
                 \> $L$ \> $R$ \\
             $U$ \> $1,0$  \> $0,1+\zeta$ \\
             $D$ \> $0,0$  \> $1,1+\zeta$
            \end{game}%
        \end{center}

        Let $i$ be the row player and $j$ be the column player.
        Since $f_j$ is dominance responsive, $f_j(G_1)(L) > f_j(G_1)(R)$, and therefore $q_i(G_1)(U) > q_i(G_1)(D)$.
        By the same argument, $f_j(G_2)(R) > f_j(G_2)(L)$, so $q_i(G_2)(D) > q_i(G_2)(U)$, and hence
        $q_i(G_1) \ne q_i(G_2)$.  But $G_1$ and $G_2$ differ only in $j$'s payoffs, so $q_i$ is other responsive.

        Similarly, $b_i(G_1)(U) > b_i(G_1)(D) = 0$, whereas $0 = b_i(G_2)(U) < b_i(G_2)(D)$, so $b_i$ is also other responsive.
    \end{proof}
\end{lemma}

\begin{theorem}
    \label{thm:fixed-point}
    Both QRE and Nash equilibrium are (profiles of) minimally strategic behavioral models.
\begin{proof}
    QRE is immediate from Lemma~\ref{lem:strategic-QBR} by letting
    $f_i(G) = QBR_i(f_{-i}(G);\lambda_i,G)$ for some set $\{\lambda_i>0\}_{i=1}^N$.
    Nash equilibrium is immediate from Lemma~\ref{lem:strategic-QBR} by letting $f$ return an arbitrary Nash equilibrium.
\end{proof}
\end{theorem}

\paragraph{Set-valued solution concepts.}
Both Nash equilibrium and QRE are set-valued functions that return sets of strategy profiles rather than a single behavior profile.  Theorem~\ref{thm:fixed-point} proves that any profile of behavioral models that maps each game to a unique Nash equilibrium (or QRE) is a profile of minimally strategic behavioral models.
Since the specific equilibrium mapped to can be arbitrary, this means that any behavioral model that is guaranteed to predict a Nash equilibrium or QRE for is minimally strategic.

\begin{theorem}
    \label{thm:no-elementary-iterative}
    All of level-$k$, cognitive hierarchy, and quantal cognitive hierarchy are minimally strategic behavioral models for agents of level 2 and higher.  Level~$1$ is minimally strategic or not depending on level~$0$.
\end{theorem}

\begin{proof}
    Let $f^{0}$ be a profile of behavioral models and $\lambda_i^{1},\lambda_i^{2}>0$ for all $i \in N$.
    To prove the result for quantal cognitive hierarchy, choose behavioral model profiles $f^{1}, f^{2}$ satisfying
    $f_i^{1}(G) = QBR_i(f_{-i}^{0}(G);\lambda_i^{1},G)$ and
    $f_i^{2}(G) = QBR_i(f_{-i}^{1}(G);\lambda_i^{2},G)$
    for all games $G$ and players $i$.

    If all of the behavioral models in $f^0$ are dominance responsive, then by the argument in the proof of Lemma~\ref{lem:strategic-QBR}, all of the behavioral models in $f^1$ are minimally strategic.

    Otherwise, all of the behavioral models in $f^1$ are dominance responsive by Lemma~\ref{lem:qbr-zeta}, and thus by the argument of the proof of Lemma~\ref{lem:strategic-QBR}, all of the behavioral models in $f^2$ are minimally strategic.

    To prove the result for level-$k$ and cognitive hierarchy, instead choose
    \[f_i^{1}(G)(a_i) = \frac{\Ind{a_i \in BR_i(f_{-i}^{0}(G))}}{\sum_{a'_i \in A_i}\Ind{a'_i \in BR_i(f_{-i}^{0}(G))}}
    \text{ and }
    f_i^{2}(G)(a_i) = \frac{\Ind{a_i \in BR_i(f_{-i}^{1}(G))}}{\sum_{a'_i \in A_i}\Ind{a'_i \in BR_i(f_{-i}^{1}(G))}},\]
    and follow an identical argument.
\end{proof}

For example, when level-$0$ is uniform, level-$1$ is strongly nonstrategic, because no change to $j$'s payoffs can change $i$'s level-$1$ behavior.  But when level-$0$ is maxmax, level-$1$ is minimally strategic.

\section{Elementary Behavioral Models}
\label{sec:elementary}

Our main task in this paper is to separate nonstrategic behavior from strategic behavior. Now that we have formally defined the latter, we can introduce a class of behavioral models, called \emph{elementary} models, that we will ultimately show are always strongly nonstrategic. 
Observe that an agent reasoning strategically needs to account both for its own payoffs (in order to be dominance responsive) and for others' payoffs (in order to be other responsive); thus, it must evaluate each outcome in multiple terms. Our key idea is thus to require that nonstrategic behavior independently ``scores'' each outcome using a single number. 
In this section, we formalize such a notion and illustrate its generality via examples of how it can be used to encode previously proposed ``nonstrategic'' behaviors. 

\subsection{Defining Elementary Behavioral Models}

The formal definition of elementary behavioral models is unfortunately more complex than the intuition we just gave. The reason is that any tuple of $k$ real values can be encoded into a single real number; in information economics this is referred to as \emph{dimension smuggling} \citeeg{nisan06}. Without ruling out dimension smuggling, therefore, a restriction that nonstrategic agents rely on only a single number would lack any force.
We thus restrict the class of functions that an elementary model can use to those that are either \emph{dictatorial} or \emph{non-encoding}; that its, functions that either ignore all inputs but one, or combine their inputs in a way that makes it impossible to determine what any single input must have been.

\begin{definition}[Dictatorial function]
A function $\varphi:\R^m \to \R^n$ is \emph{dictatorial} if its value is completely determined by a single input: $\exists i \in \{1, \ldots, m\}$ such that $\forall x,x' \in \R^{m-1}$, $\forall c \in \R$, $\varphi(x_1, \ldots, x_{i-1}, c, x_{i}, \ldots, x_{m-1}) = \varphi(x'_1, \ldots, x'_{i-1}, c, x'_{i}, \ldots, x'_{m-1})$.
\end{definition}

This class of functions takes its name from the social choice condition from which it is inspired; one input to the function $\varphi$ acts as a dictator over $\varphi$'s output.

\begin{definition}[Non-encoding function]
    A function $\varphi:\R^N \to \R$ for $N\ge2$ is \emph{non-encoding} iff for every $1\le i \le N$ and $b>0$, there exist $x,x' \in \R^N$ such that
    \begin{enumerate}
        \item $\varphi(x) = \varphi(x')$, and
        \item $\abs{x_i - x'_i} > b$.
    \end{enumerate}
\end{definition}

The reason we want such a condition is to restrict our attention to functions that combine utility values for multiple players into a single value in a meaningfully nonreversible way.
(Simple examples include summing the values, taking their max, taking the first value if it is greater than some constant and otherwise taking the second, taking a convex combination of different values, etc.)

Our condition is stronger than simply requiring the function to be non-invertible.
For a function to be non-invertible, it is sufficient that there exist a single pair of inputs in the domain that map to the same value in the range.  In contrast, we require that there be infinitely many such pairs; furthermore, there must exist pairs that are arbitrarily far apart in a specific dimension that map to the same value.
In the context of utilities that have been summarized by some non-encoding function, this means that there is no way to reliably recover the utility for a specific player based on the summary, to any degree of approximation, unless the output is always computed using \emph{only} that dimension.

For example, the linear combination $\psi^{(1)}(x) = x_1 + x_2$ is non-encoding; knowing that $\psi^{(1)}(x)=7$ gives no information about what value $x_1$ must be, nor even a non-trivial neighborhood that must contain $x_1$.
In contrast, the function
$\psi^{(2)}(x) = 10\lfloor x_1 \rfloor + \exp(x_2)/(1 + \exp(x_2))$
does not satisfy non-encoding, even though each output is mapped to by infinitely many inputs, because given $\psi^{(2)}(x)$, it is possible to approximate the value of $x_1$ to within $\pm 1$, and it is possible to recover $x_2$ exactly.

We are now ready to formally define elementary behavioral models.
Intuitively, an elementary behavioral model is one which first summarizes the outcome of each action profile as a single number computed only from the profile of utilities it induces.
The model then computes its behavior based only on these ``potentials''.

\begin{definition}[Elementary behavioral model]
    A behavioral model $f_i:\G \to \Delta^*$ is \emph{elementary} if it can be represented as $h(\Phi(G))$,
    where
    \begin{enumerate}
        \item \label{it:Phi} $\Phi(G)$ maps a game $G=(N,A,u)$ to a vector with one entry containing $\varphi(u(a))$ for each action profile $a \in A$, and
        \item $\varphi$ is either dictatorial or non-encoding, and
        \item $h: (\R^*)^* \to \Delta^*$ is an arbitrary function; we use it to map from $\R^A$ to $\Delta^{|A_i|}$.
    \end{enumerate}
    For convenience, when condition~\ref{it:Phi} holds we refer to $\Phi$ as the \emph{potential map} for $\varphi$.
\end{definition}

An elementary behavioral model works as follows. First, given an arbitrary game $G=(N,A,u)$, and for each action profile $a\in A$, we apply the same (either dictatorial or non-encoding) function $\varphi$ to each element of the $|N|$-tuple of real values $\langle u_1(a), \ldots, u_{|N|}(a)\rangle$, producing in each case a single real value. We represent all of these real values in a mapping we call $\Phi$; this \emph{potential map} is a function of the same size as each of the utility functions. We then apply an arbitrary function $h$ to $\Phi(G)$, producing a probability distribution over $A_i$. 


\subsection{Examples of Elementary Behavioral Models}
\label{sec:elementary-examples}

To demonstrate the generality of elementary behavioral models, we show how to encode each of the candidate level-$0$ behavioral models that we proposed in our past work
\citep{wright-ec}. (Thus, although that work only appealed to intuition, we can now conclude that these behavioral models are indeed all strongly nonstrategic.)

We begin with the simplest behavioral models: those that depend only on a given agent $i$'s utilities $u_i$.

\begin{example}[Maxmax behavioral model]\label{eg:maxmax}
A maxmax action for agent $i$ is an action whose best case-utility for $i$ is greater than the best-case utility of any of $i$'s other actions.
An agent that wishes to maximize their possible payoff will play a maxmax action. 
The \textit{maxmax behavioral model} $f_i^{\text{maxmax}}(G)$ uniformly randomizes over all of $i$'s maxmax actions in $G$: 
$f_i^{\text{maxmax}}(G)(a_i) \propto \Ind{a_i \in \argmax_{a'_i \in A_i} \max_{a_{-i} \in A_{-i}} u_i(a'_i, a_{-i})}.$
\end{example}

\begin{example}[Maxmin behavioral model]
A maxmin action for agent $i$ is the action with the best worst-case guarantee.  
This is the safest action to play against hostile agents.
The \textit{maxmin behavioral model} $f_i^{\text{maxmin}}(G)$ uniformly randomizes over all of $i$'s maxmin actions in $G$:
$f_i^{\text{maxmin}}(G)(a_i) \propto \Ind{a_i \in \argmax_{a'_i \in A_i} \min_{a_{-i} \in A_{-i}} u_i(a'_i, a_{-i})}.$
\end{example}

\begin{example}[Minimax regret behavioral model]\label{eg:mmr}
Following \citet{savage51}, for each action profile, an agent has a possible \emph{regret}: how much more utility could the agent have gained by playing the best response to the other agents' actions?  Each of the agent's actions is therefore associated with a vector of possible regrets, one for each possible profile of the other agents' actions.  A minimax regret action is an action whose maximum regret (in the vector of possible regrets) is minimal.  
The \textit{minimax regret behavioral model} $f_i^{\text{mmr}}(G)$ uniformly randomizes over all of $i$'s minimax regret actions in $G$. That is, if
\[r(a_i, a_{-i}) = u_i(a_i, a_{-i}) - \max_{a^*_i \in A_i} u_i(a^*_i, a_{-i})\]
is the regret of agent $i$ in action profile $(a_i, a_{-i})$, then
\[f_i^{\text{mmr}}(G)(a_i) \propto \Ind{a_i \in \argmin_{a'_i \in A_i} \max_{a_{-i} \in A_{-i}} r(a_i, a_{-i})}.\]
\end{example}

Because each of the maxmax, maxmin, and minimax regret behavioral models depends only on agent $i$'s payoffs, we can set $\varphi(u(a)) = u_i(a)$ in each case; this $\varphi$ is dictatorial.  The encodings differ only in their choice of $h$. These vary in their complexity (e.g., maxmax simply uniformly randomizes over all actions that tie for corresponding to the largest potential value; for minimax regret it is necessary to compute a best response for each action profile). However, recall that $h$ is an arbitrary function that encodes any function of the potentials.

Other behavioral models depend on both agents' utilities, and so require different $\varphi$ functions.

\begin{example}[Max welfare behavioral model]\label{eg:maxwelfare}
An \emph{max welfare} action is part of some action profile that maximizes the sum of agents' utilities. The \textit{max welfare behavioral model} $f_i^{\text{W}}(G)$ uniformly randomizes over max welfare actions in $G$: 
\[f_i^{\text{W}}(G)(a_i) \propto \Ind{a_i \in \argmax_{a'_i \in A_i} \max_{a_{-i} \in A_{-i}} \sum_{j \in N} u_j(a'_i, a_{-i})}.
\]
\end{example}

We can encode the efficient behavioral model as elementary by setting $\varphi(u(a)) = \sum_j u_j(a)$.
This $\varphi$ satisfies non-encoding; in fact, all linear combinations do.
We then define $h$ to uniformly randomize over all actions whose maximum potential value (across actions of the other player) is maximal (compared to $i$'s other actions).

\begin{proposition}
    Any linear function $\varphi(x) = w_0 + \sum_{j=1}^N w_jx_j$ is either dictatorial or non-encoding.
\end{proposition}
\begin{proof}
    If there are zero or one weights $w_j \ne 0$ with $1 \le j \ne N$, then the function is dictatorial and we are done.
    Otherwise, fix arbitrary $x \in \R^N$, $b,\epsilon>0$, and $1\le i,j \le N$ with $w_j \ne 0$ and $j \ne i$.
    Construct $x'\in\R^N$ by setting $x'_i = x_i + (1+\epsilon)b$, setting $x'_j = x_j - \frac{w_i}{w_j}(1+\epsilon)b$, and setting $x'_l = x_l$ for all $l \ne i$ with $1 \le l \le N$.
    Observe that $\varphi(x)=\varphi(x')$ and $\abs{x_i-x'_i} > b$, as required.
\end{proof}

\begin{example}[Fair behavioral model]
Let the unfairness of an action profile be the difference between the maximum and minimum payoffs among the agents under that action profile: $d(a) = \max_{i,j \in N} u_i(a) - u_j(a).$

Then a ``fair'' outcome  minimizes this difference in utilities.  A \emph{fair} action is part of a minimally unfair action profile. The \textit{fair behavioral model} $f_i^{\text{fair}}(G)$ uniformly randomizes over fair actions:
$f_i^{\text{fair}}(G)(a_i) \propto \Ind{a_i \in \argmin_{a'_i \in A_i} \min_{a_{-i} \in A_{-i}} d(a'_i, a_{-i})}.
$
\end{example}

We can encode the fair behavioral model as elementary by setting $\varphi(u(a)) = \max_{j,k} (u_j(a)-u_k(a))$; it is again straightforward to demonstrate that this potential function is non-encoding.
We then define $h$ to uniformly randomize over all actions whose minimum potential value is minimal.

\begin{proposition}
    Let $\varphi:\R^N\to\R$ for some $N \ge 2$ be defined by $\varphi(x) = \max_{1 \le j,k\le N}(x_j - x_k)$.
    Then $\varphi$ is non-encoding.
\end{proposition}
\begin{proof}
    Fix arbitrary $x \in \R^N$, $b,\epsilon>0$.
    Let $x'_i = x_i + (1+\epsilon)b$ for every $1 \le i \le N$.
    Clearly $\abs{x_i-x'_i} > b$ for every $1 \le i \le N$.
    Since $x'_j - x'_i = x_j - x_i$ for every $1 \le i,j \le N$, we also have
    $\varphi(x) = \varphi(x')$.
\end{proof}

Finally, we note that all of the examples just given are binary: actions are either fair/maxmin/etc.\ or they are not. By changing only $h$, we could similarly construct continuous variants of each concept in which, e.g., actions that achieve nearly maximal potentials are played nearly as often by the behavioral model.


\section{Elementary Behavioral Models are Strongly Nonstrategic}
\label{sec:characterize}

We are now ready to show that elementary behavioral models are always strongly nonstrategic. 
This result is important because it achieves our key goal of distinguishing strategic from nonstrategic behavioral models via a formal mathematical criterion, rather than relying on the intuitive sense that a model ``depends on'' an explicit model of an opponent's behavior. In fact, we do a bit better than simply showing that elementary models are nonstrategic: we show that the space of dominance responsive behavioral models is exactly partitioned into elementary models and minimally strategic models.

\begin{theorem}
    \label{thm:elementary-nonstrategic}
    Every elementary behavioral model is strongly nonstrategic.
\end{theorem}
\begin{proof}
    Suppose for contradiction that elementary behavioral model $f_i(G) = h_i(\Phi(G))$ is minimally strategic, where $\Phi$ is the potential map for $\varphi$.
    By the definition of elementary behavioral models, $\varphi$ is either dictatorial or non-encoding.
    \begin{enumerate}
        \item\label{case:i-dict} \emph{$i$ is a dictator for $\varphi$.}
        Because $f_i$ is other responsive, there exist $G=(N,A,u)$ and $G'=(N,A,u')$ with $u_i(a) = u'_i(a)$ for all $a \in A$, such that $f_i(G) \ne f_i(G')$.   But since $i$ is a dictator for $\varphi$, $\Phi(G) = \Phi(G')$, and hence $f_i(G) = f_i(G')$, a contradiction.

        \item\label{case:not-dict} \emph{$\varphi$ is non-encoding.}
        Since $f_i$ is assumed to be minimally strategic, it must be $\zeta$-dominance responsive for some $\zeta$.
        Let $x,x' \in \R^2$ be two vectors such that $x'_i - x_i > \zeta$.
        These vectors are guaranteed to exist by the definition of non-encoding.
        We use these utility vectors to construct $2$-player games $G_3$ and $G_4$:
        \begin{center}
            \begin{game}{3}{2}[$G_3$]
                 \> $L$ \> $R$ \\
             $U$ \> $x'$  \> $x'$ \\
             $D$ \> $x$  \> $x$
            \end{game}%
            \hspace{1cm}%
            \begin{game}{3}{2}[$G_4$]
                 \> $L$ \> $R$ \\
             $U$ \> $x$  \> $x$ \\
             $D$ \> $x'$  \> $x'$
            \end{game}%
        \end{center}
        Note that $i$'s utility in $G_3$ for playing $U$ exceeds that of playing $D$ by $\zeta$, and thus by $\zeta$-dominance responsiveness $f_i(G_3)(U) > f_i(G_3)(L)$.
        By the same argument, $f_i(G_4)(D) > f_i(G_4)(U)$.  Thus $f_i(G_3) \ne f_i(G_4)$.
        But since $\varphi(x)=\varphi(x')$, and since $x$ and $x'$ are the only payoff tuples that occur in either $G_3$ or $G_4$, $\Phi(G_3) = \Phi(G_4)$ and hence $f_i(G_3) = f_i(G_4)$, a contradiction.\qedhere

        \item\label{case:j-dict} \emph{$j\ne i$ is a dictator for $\varphi$.}
        Since $\varphi(x)=\varphi(x')$ whenever $x_j=x'_j$,
        it is straightforward to construct $x,x' \in \R^2$ such that $x_i-x'_i > \zeta$ and $\varphi(x)=\varphi(x')$.
        But then we can construct $G_3$ and $G_4$ and derive a contradiction as in case~\ref{case:not-dict}.
    \end{enumerate}       
\end{proof}

For example, the maxmax, maxmin, and minimax regret behavioral models from Section~\ref{sec:elementary-examples} all have a dictatorial potential, and therefore fail to be other responsive.  This is easy to see, since changes in the utilities of the other players do not change a dictatorial potential's value at all.
In contrast, the max welfare behavioral model is other responsive, since changing just the utilities of the other agents can change the sum of utilities.  However, it is not dominance responsive.  No matter how dominant an action is, the sum of utilities for its associated outcomes can be entirely arbitrary, depending on the utilities of the other player.  The fair behavioral model is similarly other responsive (changing just the other agents' utilities can also change the differences), but not dominance responsive (since an action's being dominant for one player does not imply anything specific about the differences between agent utilities that the fair model must base its predictions upon).

The set of elementary behavioral models includes every strongly nonstrategic dominance responsive model.

\begin{theorem}
    \label{thm:DR-partition}
    A dominance responsive model is other responsive iff it is not elementary.
\end{theorem}
\begin{proof}
    \emph{Only-if direction: dominance responsive and not elementary implies not other responsive.}
    If a model is both dominance responsive and elementary, then by Theorem~\ref{thm:elementary-nonstrategic}, it is not other responsive.

    \emph{If direction: If a model is dominance responsive and not other responsive, then it is elementary.}
    Suppose that a behavioral model $f_i$ is dominance responsive, but not other responsive.
    Therefore, for every pair of games $G=(N,A,u)$ and $G'=(N,A,u')$ with $u_i(a) = u'_i(a)$ for all $a \in A$, $f_i(G) = f_i(G')$.
    We show how to represent $f_i$ as an elementary function by constructing appropriate $\varphi$, $\Phi$, and $h_i$ functions.
    Define $\varphi(x) = x_i$. Let $h_i(\Phi(G)) = f_i(z(\Phi(G)))$, where $z:\R^A \to \G$ is a function that returns a game with the utilities of $i$ set to its argument and the utilities of the other players set to 0.  Since differences in the other agents' utilities never change the output of $f_i$,
    $h_i(\Phi(G)) = f_i(z(\Phi(G))) = f_i(G)$ for all $G$.
\end{proof}


\section{Combinations of Elementary Models}
\label{sec:elementary-combos}
We now consider behavioral models that are constructed by combining the predictions of multiple elementary models.  We begin by demonstrating that the strong nonstrategic property is not preserved by convex combinations of elementary models; it is possible to construct a minimally strategic behavioral model by averaging the predictions of just two elementary models.

However, it is not possible to build arbitrary models by convex combinations of elementary models.
We demonstrate that such convex combinations cannot represent any of the exemplar strategic models of section~\ref{sec:analysis-existing}, and define ``weakly nonstrategic'' models to be those that lack this ability.  Finally, we demonstrate that convex combinations of elementary models augmented with a filtering preprocessing operation are also weakly nonstrategic in this sense.

\subsection{Convex Combinations of Elementary Models}
\label{sec:convex-elementary}
We call a model that returns a weighted average of the predictions of a set of elementary models a convex combination of elementary models.

\begin{definition}[Convex combination of elementary models]
    \label{def:convex-elementary}
    A behavioral model $g_i$ is a \emph{convex combination of elementary models} when there exists
    a set of $K$ elementary models $\fk$ and a set of $K$ weights $w_k \in [0,1]$ with $\sum_{k=1}^K w_k = 1$ such that for all games $G=(N,A,u)$ and actions $a_i \in A_i$,
    \[g_i(G)(a_i) = \sum_{k=1}^K w_kf_i^{(k)}(G)(a_i).\]
\end{definition}

Convex combinations do not automatically preserve the strong nonstrategic property.
As a simple example, the behavioral model
\[ g_i(G)(a_i) = \frac{1}{3} f_i^{\text{W}}(G)(a_i) + \frac{2}{3} f_i^{\text{maxmax}}(G)(a_i) \]
that averages the predictions of the max-welfare and maxmax behavioral models is both dominance responsive and other responsive. It will play any strictly dominant action with probability at least 2/3, but the 1/3 probability mass allocated by $f_i^{\text{W}}$ can differ in games that differ only in the other players' payoffs.

\subsection{Weakly Nonstrategic Models}
\label{sec:weakly-nonstrategic}
Although convex combinations of elementary models can be minimally strategic, they are not arbitrarily expressive.  In this section, we demonstrate that this class of model is \emph{weakly nonstrategic} in a specific sense. 

\begin{definition}
    \label{def:weakly-nonstrategic}
    A class $F$ of behavioral models is \emph{weakly nonstrategic} if it is not able to represent logit quantal best response to a dominance responsive behavioral model.
\end{definition}

This condition is sufficient to rule out the representation of any of the example strategic models in section~\ref{sec:analysis-existing}.  First, we show that this condition holds for convex combinations of elementary models.

\begin{theorem}
    \label{thm:cvx-combo-weakly-nonstrategic}
    For any $\lambda>0$ and dominance responsive profile of behavioral models $f_{-i}$ and any convex combination of elementary models $g_i$, there exists a game $G$ for which $g_i(G) \ne q_i(G)$, where
    \[q_i(G) = QBR_i(f_{-i}(G); \lambda, G).\]    
\end{theorem}
\begin{proof}
    We begin by observing that logit response satisfies a stronger condition than dominance responsiveness.
    For sufficiently large differences of utility, strictly dominant actions are played not just with greater probability than other actions, but in fact with probabilities that can be made arbitrarily close to 1.
    No convex combination that includes an other responsive elementary model can satisfy this condition.

    To see this, let $g_i = \sum_{k=1}^K f_i^{(k)}$ be a convex combination of elementary models, and let $f_i^{(l)}$ be other responsive with potential $\phi^{(l)}$.  Let $x,x'\in\R^2$ be two vectors such that $\phi^{(l)}(x) = \phi^{(l)}(x')$ and
    \[ x_i' - x_i > \frac{1}{\lambda}\log\left(\frac{1-w_l/2}{w_l/2}\right). \]
    Now consider games $G_5$ and $G_6$:
    \begin{center}
            \begin{game}{3}{2}[$G_5$]
                 \> $L$ \> $R$ \\
             $U$ \> $x'$  \> $x'$ \\
             $D$ \> $x$  \> $x$
            \end{game}%
            \hspace{1cm}%
            \begin{game}{3}{2}[$G_6$]
                 \> $L$ \> $R$ \\
             $U$ \> $x$  \> $x$ \\
             $D$ \> $x'$  \> $x'$
            \end{game}%
        \end{center}
        In these games, $q_i(G_5)(U) = q_i(G_6)(D) > 1 - w_l/2$, since
        \begin{align*}
            q_i(G_5)(U) &= \frac{\exp[\lambda u^{(5)}_i(U, f_{-i}(G_5))]}
                                {\exp[\lambda u^{(5)}_i(U, f_{-i}(G_5))] + \exp[\lambda u^{(5)}_i(D, f_{-i}(G_5))]} \\
                        &= \frac{1}{1+\exp[\lambda(u^{(5)}_i(D f_{-i}(G_5)) - u^{(5)}_i(U, f_{-i}(G_5)))]} \\
                        &> \frac{1}{1+\exp[-\lambda(1/\lambda)\log\left(\frac{1-(w_l/2)}{(w_l/2)}\right)]}\\
                        &= \frac{1}{1+\frac{(w_l/2)}{1-(w_l/2)}}\\
                        &= \frac{1}{\frac{1-(w_l/2)+(w_l/2)}{1-(w_l/2)}} \\
                        &= 1-w_l/2.
        \end{align*}
        Without loss of generality, assume that $f_i^{(l)}(G_5)(U) > 1/2$.
        Since $\Phi^{(l)}(G_5) = \Phi^{(l)}(G_6)$, it must be that $f_i^{(l)}(G_6)(U) > 1/2$ also.
        But then $f_i^{(l)}(G_6)(D) < 1/2$, and therefore $g_i(G_6)(D) < 1-w_l/2$.

        Of course, since we know from Lemma~\ref{lem:strategic-QBR} that $q_i$ is other responsive, $g_i$ also cannot represent $q_i$ without including at least one other responsive model, and we have our result.
\end{proof}

\subsection{Convex Combinations with Filtering}
\label{sec:convex-filtered-elementary}
In earlier work \citep{wright-ec,wright-jair}, we proposed a purportedly nonstrategic level-$0$ model that consisted of a convex combination of elementary behavioral models, with an additional \emph{informativeness filtering} step.  Models of this kind discard uniform predictions by sub-models as ``uninformative''; instead, for a given game, the model's prediction is a convex combination of the predictions of those sub-models that made an ``informative'', non-uniform prediction.  In this section, we formally define a generalization of these models, and demonstrate that models we proposed are indeed weakly nonstrategic in the sense of Definition~\ref{def:weakly-nonstrategic}.

\begin{definition}[Filtered convex combination of elementary models]
    A behavioral model $g_i$ is a \emph{filtered convex combination of elementary models} when there exist a set of $K$ elementary models $f_i^{(k)}$, $K$ weights $w_k \in [0,1]$ with $\sum_{k=1}^K w_k = 1$, and a \emph{filtering function} $h:\Delta^*\to\{0,1\}$ such that for all games $G=(N,A,u)$ and actions $a_i \in A_i$,
    \begin{equation}
        \label{eq:filtered-convex}
       g_i(G)(a_i) = \begin{cases}
                         1 / |A_i| &\text{if } h(f_i^{(k)}(G))=0 \quad\forall 1\le k\le K,\\
                         \frac{\sum_{k=1}^K h(f_i^{(k)}(G))w_kf_i^{(k)}(G)(a_i)}{\sum_{k=1}^K h(f_i^{(k)}(G))w_k}
                         &\text{otherwise.}
                     \end{cases} 
    \end{equation}
\end{definition}
Rather than returning a weighted average of the predictions of all its sub-models, a filtered convex combination of elementary models returns a weighted average of the predictions of the sub-models whose predictions are admitted by the filtering function (or a uniform distribution if none of the sub-models' predictions are admitted).  A natural choice of filtering function is to admit only non-uniform predictions, but this definition allows for any filter that is a function only of a sub-model's prediction.

The addition of a filtering function turns out to be insufficient to allow convex combinations of elementary models to represent our strategic behavioral models, when the elementary models in question have a potential that satisfies a stronger version of non-encoding.

\begin{definition}[Uniformly non-encoding function]
    A function $\varphi:\R^N \to \R$ for $N\ge2$ is \emph{uniformly non-encoding} iff for every $x \in \R^N$, there exists $x' \in \R^N$ such that
    \begin{enumerate}
        \item $\varphi(x) = \varphi(x')$, and
        \item $x_i - x'_i > b$.
    \end{enumerate}
\end{definition}
That is, not only does there exist a pair $x,x'$ that are arbitrarily far apart in a given dimension, but for \emph{every specific $x$} there exists an $x'$ that is arbitrarily far from it in a given dimension.

The elementary models proposed in \citet{wright-ec,wright-jair} all had linear potential functions.
It turns out that all non-dictatorial linear functions satisfy this condition.
\begin{proposition}
    Any linear function $\varphi(x) = w_0 + \sum_{j=1}^N w_jx_j$ is either dictatorial or non-encoding.
\end{proposition}
\begin{proof}
    If there are zero or one weights $w_j \ne 0$ with $1 \le j \ne N$, then the function is dictatorial and we are done.
    Otherwise, fix arbitrary $x \in \R^N$, $b,\epsilon>0$, and $1\le i,j \le N$ with $w_j \ne 0$ and $j \ne i$.
    Construct $x'\in\R^N$ by setting $x'_i = x_i + (1+\epsilon)b$, setting $x'_j = x_j - \frac{w_i}{w_j}(1+\epsilon)b$, and setting $x'_l = x_l$ for all $l \ne i$ with $1 \le l \le N$.
    Observe that $\varphi(x)=\varphi(x')$ and $x_i'-x_i > b$, as required.
    Since $x$ was arbitrary, this pair can be found for any $x$.
\end{proof}



\begin{theorem}
\label{thm:filtered-convex-universal-nonencoding-weakly-nonstrategic}
    Let $g_i = \sum_{k=1}^K w_kf_i^{(k)}$ be a filtered convex combination of elementary models, and let
    $q_i(G) = QBR_i(f_{-i}(G);\lambda,G)$ for any $\lambda>0$ and dominance responsive profile of behavioral models $f_{-i}$.  If each potential $\varphi^{(k)}$ is either dictatorial or uniformly non-encoding,
    then there exists a $2\times2$ game $G$ for which $g_i(G) \ne q_i(G)$; that is, $g_i$ is weakly nonstrategic.
\end{theorem}
\begin{proof}
    Let $f_i^(l)$ be an other responsive model in the convex combination.  If no such model exists, then $g_i$ is not other responsive and we are done.

    Fix an arbitrary game $G_7=(\{i,j\}, \{U,D\}\times\{L,R\}, u^{(7)})$ such that $h(f_i^{(l)}(G_7)) = 1$.
    There must be some such game for some other responsive model in the convex combination; otherwise, $g_i$ is not other responsive for $2\times2$ games, and we are done, since $q_i$ is other responsive for $2\times2$ games.

    Suppose without loss of generality that $f_i^{(l)}(G_7)(U) > 1/2$.
    We now use a construction similar to that of Theorem~\ref{thm:cvx-combo-weakly-nonstrategic} to build a game $G_8$ for which $q_i(G_8)(D) > 1 - w_l/2$, but $g_i(G_8)(D) < 1 - w_l/2$.

    Let $x = u^{(7)}(D,L)$ and $y = u^{(7)}(D,R)$. Choose vectors $x',y'$ 
    such that
    \begin{align*}
        \varphi(x') &= \varphi(x), \\
        \varphi(y') &= \varphi(y), \\
        x_i' - x_i  &> \frac{1}{\lambda}\log\left(\frac{1-w_l/2}{w_l/2}\right) + u_i(U,L),  
        y_i' - y_i  &> \frac{1}{\lambda}\log\left(\frac{1-w_l/2}{w_l/2}\right) + u_i(U,R).
    \end{align*}

    Now construct $G_8$ as follows:
        \begin{center}
            \begin{game}{3}{2}[$G_8$]
                 \> $L$ \> $R$ \\
             $U$ \> $u^{(7)}(U,L)$  \> $u^{(7)}(U,R)$ \\
             $D$ \> $x'$  \> $y'$
            \end{game}%
     \end{center}

     Observe that $u^{(8)}_i(D,f_{-i}(G_8)) - u^{(8)}_i(U,f_{-i}(G_8)) > \frac{1}{\lambda}\log\left(\frac{1-w_l/2}{w_l/2}\right)$ for any value of $f_{-i}(G_8)$,
     and thus $q_i(G_8)(D) > 1-w_l/2$.  But since $\varphi(x') = \varphi(x)$ and $\varphi(y') = \varphi(y)$,
     $f_i^{(l)}(G_8)(U) = f_i^{(l)}(G_7)(U) > 1/2$, and so $g_i(G_8)(D) < 1 - w_l/2$, giving our result.
\end{proof}

\begin{corollary}
    Any filtered convex combination of elementary models with linear potentials is weakly nonstrategic.
\end{corollary}

\section{Discussion and Future Work}
\label{sec:discussion}

In this work, we proposed elementary behavioral models (and their convex combinations) as mathematical characterizations of classes of nonstrategic decision rules.  These classes are constructively defined, in the sense that membership of a rule is verified by demonstrating how to represent the rule in a specific form---as a function of the output of a non-encoding potential map---rather than by proving that it cannot be represented as a response to probabilistic beliefs.
Indeed, many apparently nonstrategic rules \emph{can} be represented as responses to probabilistic beliefs; what is crucial to our notion of nonstrategic behavior is that they \emph{need not} be.

How to formally define nonstrategic behavior is a very tricky question, and we do not commit to a complete characterization.  Instead, we define two notions. First, strongly nonstrategic models are those that do not meet even the minimal requirements of dominance responsiveness and other responsiveness.  We claim that any reasonable definition of strategic behavior must imply both of these requirements, and hence any model that fails either requirement is unambiguously (``strongly'') nonstrategic.  Second, weakly nonstrategic classes of models are those that do not include any of our example notions of strategic models.  In the previous section, we describe some models classes that are weakly nonstrategic even though some of their members are minimally strategic.  A more complete characterization of nonstrategic behavior that classifies these ``in-between'' model classes is an important direction for future work.

It is interesting to note that various special cases of strategic solution concepts are nonstrategic under our definition.
For example, the equilibrium of a two-player zero-sum game can be computed by considering only the utility of a single agent, and hence the behavior for an equilibrium-playing player in such a game can be computed by an elementary behavioral model that computes the agent's maxmin strategy.\footnote{However, note that such a behavioral model would act in every game as though the game was zero sum; it is in this sense that we would still say that the model is nonstrategic. An analogous caveat applies to the other examples we give here.} Similarly, an equilibrium for a potential game can of course be computed in terms of outcome values computed by a potential function \citep{monderer96}.

One thing that these exceptions all have in common is that they are also \emph{computationally} easy, unlike general $\epsilon$-equilibrium, which is known to be hard in a precise computational sense \citep{daskalakis09,chen06}.
The equilibrium of a zero-sum game can be solved in polynomial time by a linear program; 
the equilibrium of a potential game can be found simply by finding the maximum of the potential function over all pure outcomes.\footnote{\label{fn:potential-games}
Note that although this procedure is tractable in terms of the normal form, congestion games (the most important representation of potential games) can often be represented in an asymptotically more compact form than the normal form; finding the equilibrium of such games is intractable relative to the compact representation \citep{fabrikant04,babichenko21}.}

However, the connection between ease of computation and strategic simplicity is not an equivalence.
An attractive future direction with potential applications in the design and analysis of multiagent environments is to shed further light on the connection between computational and strategic simplicity.

We also observe that our characterization of nonstrategic behavior in this paper is a binary distinction: in the view we have advanced, a behavioral model is either nonstrategic or it is not.
An intriguing question for future work is whether such a distinction can be made more quantitative: i.e., is there a sense in which agents are nonstrategic to a greater or lesser degree that is distinct from the number of steps of strategic reasoning that they perform?

\section*{Acknowledgments}
We are grateful to several anonymous reviewers and the associate editor for their invaluable comments and suggestions.  This work was funded in part by an NSERC E.W.R.\ Steacie Fellowship and an NSERC Discovery Grant.  Part of this work was done at Microsoft Research New York while the first author was a postdoctoral researcher and the second author was a visiting researcher.  Both authors held Canada CIFAR AI Chairs through the Alberta Machine Intelligence Institute during part of this work.

\clearpage
\bibliographystyle{apalike}
\bibliography{nonstrategic,lit}
\end{document}